\title{A Reputation for Honesty\footnote{We thank Mehmet Ekmekci and Navin Kartik for helpful comments, and National Science Foundation grants SES-1947021 and SES-1951056 for financial support.}}
\author{Drew Fudenberg\footnote{Department of Economics, Massachusetts Institute of Technology. Email: drewf@mit.edu} \and Ying Gao\footnote{Department of Economics, Massachusetts Institute of Technology. Email: yingggao@gmail.com} \and Harry Pei\footnote{Department of Economics, Northwestern University. Email: harrydp@northwestern.edu}}
\date{\today}

\documentclass[12pt]{article}
\usepackage{amsmath}
\usepackage{graphicx}
\usepackage{amsfonts}
\usepackage{amsthm}
\usepackage{setspace}
\usepackage{tikz}
\usepackage{amssymb}
\usetikzlibrary{patterns}
\usepackage{sgame}
\usepackage{color}
\usepackage{hyperref}
\usepackage{times}
\usepackage{enumitem}
\usepackage{mathptmx}
\usepackage{comment}
\usepackage[round]{natbib}
\usepackage[bottom]{footmisc}
\usepackage[top=1in, bottom=1in, left=0.9in, right=0.9in]{geometry}

\begin{document}
\maketitle
\numberwithin{equation}{section}
\noindent We analyze situations in which players build reputations for honesty rather than for playing particular actions.
A patient player facing a sequence of short-run opponents makes an announcement about their intended action after observing an idiosyncratic shock, and  before players act. The patient player is either an honest type whose action coincides with their announcement, or an opportunistic type who can freely choose their actions. We show that the patient player can secure a high payoff by building a reputation for being honest
when the short-run players face uncertainty about which of the patient player's actions are currently feasible, but may receive a low payoff when there is no such uncertainty.  
\\


\begin{spacing}{1.5}

\newtheorem{Proposition}{\hskip\parindent\bf{Proposition}}
\newtheorem{Theorem}{\hskip\parindent\bf{Theorem}}
\newtheorem{Lemma}{\hskip\parindent\bf{Lemma}}
\newtheorem{Corollary}{\hskip\parindent\bf{Corollary}}
\newtheorem{Definition}{\hskip\parindent\bf{Definition}}
\newtheorem{Assumption}{\hskip\parindent\bf{Assumption}}
\newtheorem*{Condition}{\hskip\parindent\bf{Supermodularity Condition}}
\newtheorem{Claim}{\hskip\parindent\bf{Claim}}

Many economic actors have reputations for keeping or breaking their promises.   As a prominent example, Archibald Cox Jr's default on his promise of paying high bonuses in the early 90s triggered a massive defection of key personnel from  First Boston to its archrival Merrill Lynch.\footnote{ See ``Taking the Dare'' in The New Yorker, July 26th, 1993. The departed individuals include leaders of First Boston's prestigious energy group, and more than a dozen managing directors in its fixed-income and mortgage-backed-securities groups, triggered by a lower bonus payment than what they had been promised. Many who departed had been at First Boston their entire careers, including during its difficult times in the 80's.} 
Similar logic applies to advertising and marketing, which can set customers' expectations about the types of interactions they are going to have with the firm. If those expectations are not aligned with the actual customer experience, the firm's brand and business will suffer.

Motivated by these observations, we examine the reward for building  a  reputation for honesty. Compared to reputations for taking specific actions, a reputation for honesty can better adapt an agent's decisions to the current circumstances, which is valuable when the environment changes over time. Moreover, it is unrealistic to make commitments based on future contingencies that are hard to describe in advance, and the simplicity of a commitment to honesty makes it more plausible.


In our model, a patient player (e.g., a firm) faces a sequence of myopic opponents (e.g., consumers), each of whom  plays the game only once. Each period, before players act, the patient player privately observes an idiosyncratic shock,
which can affect their payoff 
(e.g., their production cost)
and which of their actions are currently feasible. Then the patient player
announces the action they intend to play. 
The myopic players cannot observe the shocks, but can observe the announcement in the current period as well as  whether the patient player has kept their word in the past.

The patient player is either an \textit{honest type}, who strategically chooses their announcements but always keeps their word, or an \textit{opportunistic type}, who strategically chooses both the announcements and the actions. Both types have the same payoff function. This contrasts to \cite{krepsWilson}, \cite{milgromRoberts}, and \cite{fl89} in which with positive probability, the patient player is a commitment type who mechanically plays a particular action.

Theorem \ref{Theorem1} shows that  the patient player receives at least their expected Stackelberg payoff in every equilibrium 
when the myopic players face a small amount of uncertainty about the actions currently available to the patient player.\footnote{In Section \ref{sec4}, we show that our reputation result extends when the patient player observes which of their actions are feasible after making their announcement, or when the patient player chooses an action (e.g., their effort), observes their product quality, and makes an announcement about quality before the myopic player chooses their action.} 
A complication  is that the opportunistic type may announce certain actions with higher probability than the honest type does, so the patient player's announcement may adversely affect their opponent's belief about their type. As a result, both types of the patient player may face a tradeoff between announcing actions 
that lead to higher credibility and announcing actions that lead to higher commitment payoffs (i.e., payoff conditional on being trusted).

To see why the reputation bound nevertheless obtains,  suppose the honest type  announces their Stackelberg action whenever it is feasible. When a myopic player does not best reply against the announcement, whether the patient player keeps their word in that period is informative about their type.
Because the set of feasible actions is stochastic,  the honest type announces each action with strictly positive probability, which implies that observing the current announcement leads to at most a bounded change in the myopic player's belief. 
Therefore, when the patient player behaves honestly, there can be at most a bounded number of periods in which the myopic players do not best reply to the announcement. As a result, 
the patient player receives at least their expected Stackelberg payoff.

By contrast, Theorem \ref{Theorem2} shows that when the patient player can choose from any of their possible actions in every period, there are equilibria in which they receive a low payoff, which can be as low as their minmax value in examples such as the product choice game.  

\paragraph{Related Literature:} Our paper contributes to the study of reputation models 
 where no types  are committed to specific actions.
\cite{schmidt} characterizes the Markov equilibria of  finite-horizon repeated bargaining games in which a firm has private information about its production cost. 
\cite{pei} characterizes an informed player's highest Nash equilibrium payoff when facing uninformed opponents. 
\cite{suwo} constructs a cooperative equilibrium  in a community enforcement model with a type that  communicates strategically but is committed to playing \textit{always defect}. By contrast, we provide a lower bound on the patient player's payoff for all Nash equilibria.

Our reputation result requires the uninformed players to face uncertainty about the availability of the informed player's actions, or more generally, believe that
the honest type 
makes every announcement with positive probability. This is related to \citet*{celentaniEtAl} and \cite{atakhanEkmekci},
which  show that full support monitoring can help reputation building when the uninformed player is long-lived. Their results, unlike ours,  require that
the informed player cannot perfectly observe the uninformed player's actions.

 \cite{jullienEtAl20} studies repeated buyer-seller games in which a seller privately observes their product quality, which is a noisy signal of their effort.
It shows that cheap talk communication about quality improves the maximum social welfare  if and only if the seller's cost of effort is intermediate.\footnote{
Jullien and Park (2014) shows that communication accelerates consumer learning when product quality is determined by the seller's type, and the high type seller is non-strategic and always tells the truth. 
Awaya and Krishna (2016) identifies a class of games in which players can achieve  perfectly collusive payoffs with  communication, but not  without it.} 
 Our paper examines a complementary question, namely, whether a patient player can guarantee high payoffs in \textit{all} equilibria by building reputations for honesty.
Successful reputation building in our model  
requires uncertainty about the actions available to the patient player, 
but does not depend on the players' payoff functions. Corollary \ref{Cor2} in Section 4 extends our insights to Jullien and Park (2020)'s setting, which implies that a patient seller receives their optimal commitment payoff in all equilibria when product quality (i.e., the seller's private signal) is a noisy signal of effort, but receives a payoff lower than that in some equilibria when quality is a perfect signal of effort.

The fact that many people prefer to be honest has been established experimentally by e.g.  \cite{gneezy05} and \citet*{gneezyEtAl18}. \citet*{kartikEtAl07} and \cite{kartik09} show how  costs of lying
change the equilibrium outcomes of strategic communication games.
Instead of positing that some players have  a cost of lying, we follow \citet*{CKS08} and \cite{C11} and assume that the patient player is either an honest type who never lies, or an opportunistic type who  faces no cost of lying. 
Our results extend to cases with strictly positive and possibly heterogeneous lying costs. 

Our work is related to the literature on pre-play communication. 
Including an honest type in our model is in line with the experimental finding of \cite{charness} that some people keep their word in order to live up to others’ expectations.
\cite{Sobel} allows one of the players to communicate their intended action before playing a two-player complete information game, and provides sufficient conditions under which the sender receives their highest Nash equilibrium payoff.

\section{Example: Product Choice Game with Stochastic Cost}\label{sub2.1}
Consider a game between a firm (row player) with discount factor $\delta \in (0,1)$ and a sequence of consumers (column player), each of whom plays the game only once. In every period, the firm privately observes its i.i.d. cost of production $\theta_t \in \{\theta_g,\theta_b\}$. 
Let $p_g \in (0,1)$ be the probability that
$\theta_t=\theta_g$. The players' stage-game payoffs are:
\begin{center}
\begin{tabular}{| c | c | c |}
  \hline
  $\theta=\theta_g$ & $T$ & $N$ \\
  \hline
  $H$ & $1,2$ & $-1,0$ \\
  \hline
  $L$ & $2,-2$ & $0,0$ \\
  \hline
\end{tabular}
\quad
\begin{tabular}{| c | c | c |}
  \hline
  $\theta=\theta_b$ & $T$ & $N$ \\
  \hline
  $H$ & $-1,2$ & $-3,0$ \\
  \hline
  $L$ & $2,-2$ & $0,0$ \\
  \hline
\end{tabular}
\end{center}
When $\theta=\theta_g$, the best pure-strategy commitment for the firm is  to action $H$, which yields  payoff $1$.\footnote{A commitment to a mixed strategy is even better in state $\theta_g$. 
We do not consider reputations for playing mixed actions in this paper.} 
When $\theta=\theta_b$, the firm's optimal commitment action is $L$, which yields  payoff $0$.
If the firm obtains its optimal pure-strategy  commitment payoff
in every state, then its expected payoff is $p_g$.

\paragraph{No Announcement Benchmark:} Suppose the firm cannot 
make announcements about its intended actions, and that 
with small but positive probability it is a commitment type that mechanically plays $H$ in every period. 
Future consumers can observe the firm's effort in previous periods, but not the past realizations of $\theta_t$.\footnote{This is a reasonable assumption given that $\theta$ only affects the firm's cost of supplying high quality.}  Then, there are  equilibria in which the patient firm's payoff is $\max\{0, 2p_g-1\}$, which is strictly lower than $p_g$. For example, when $p_g \geq 1/2$, there is an equilibrium where the  opportunistic firm chooses $H$ in every period on the equilibrium path, and each consumer chooses $T$ unless they observe $L$ in at least one of the previous periods. 
Intuitively, when $\theta=\theta_b$, the cost of playing $H$ outweighs the benefit from the consumer's trust, and 
the firm faces a tradeoff between sustaining its reputation for playing $H$ and avoiding the excessive cost.

This low-payoff equilibrium motivates our interest in reputations for honesty.

\paragraph{Reputation for Honesty:} Suppose that the firm can make an announcement $m_t$ about its intended action $a_t$ to the current consumer after observing $\theta_t$, but before players choosing actions.

The firm is either honest or opportunistic. In contrast to the commitment types in canonical reputation models, the honest type is strategic when making announcements and does not commit to any particular action. Instead, it commits to play the action it announces in every period. The two  types of the firm have the same stage-game payoff function and discount factor, and maximize their respective discounted average payoffs.
The consumer in period $t$ observes the firm's announcement in period $t$, as well as the value of $\mathbf{1}\{a_s=m_s\}$ for $s \in \{0,1,...,t-1\}$, i.e., 
whether the firm's announcements matched its actions in the previous periods.

As Theorem 2 shows,  the firm's equilibrium payoff can be low when all of its actions are always available. To see how this works in the example, consider the following strategy profile: Both types of the firm announce $L$ and play $L$ at every  history, and each consumer plays $N$ regardless of the firm's announcement. 
The consumers' belief about the firm's type never changes on the equilibrium path. 
After the firm announces $H$, 
the current consumer believes 
that the firm is opportunistic and will play $L$.\footnote{When future consumers only observe whether $a_t$ coincides with $m_t$, but not the exact realizations of $a_t$ and $m_t$, they do not  observe deviations in the announcement stage if the firm kept its word. We show that the firm can also receive a low payoff when future consumers can observe both $a_t$ and $m_t$.} 
This strategy profile and assessment constitute a Perfect Bayesian equilibrium, in which the firm's discounted average payoff is $0$ regardless of its type.

This low-payoff equilibrium is driven by the honest-type firm's strategic concerns when making announcements. The consumers believe that the opportunistic type is more likely to announce $H$, so the honest type faces a trade-off in state $\theta_g$ between announcing an action that leads to higher credibility (i.e., action $L$) and an action that leads to a higher commitment payoff (i.e., action $H$). This motivates the honest type to announce $L$, making consumers' beliefs self-fulfilling.


In contrast, Theorem 1 shows that when some of the firm's actions are unavailable with small but positive probability, both types of the firm 
receive at least their expected Stackelberg payoff in every equilibrium. 
For example, suppose in every period the firm can choose between $H$ and $L$ with probability $1-2\varepsilon$, can only choose $H$ with probability $\varepsilon$, and can only choose $L$ with probability $\varepsilon$. 
Both types of the patient firm receive payoff at least 
$(1-\varepsilon) p_g - 7 \varepsilon (1-p_g)$
when the feasibility of actions is i.i.d. over time and is independent of $\theta$.
This guaranteed payoff converges to $p_g$ as 
 $\varepsilon \rightarrow 0$.

\paragraph{Remarks:} Our assumption that the consumers face uncertainty about which of the firm's actions are feasible fits situations in which the firm is a single contractor who  occasionally may be sick, and so unable to provide high-quality service. 
It also fits cases where the firm faces occasional regulatory inspections, and producing low-quality products in those periods can lead to fines and the risk of being shut down. In this situation, the firm will always choose to supply high quality, regardless of their discount factor.\footnote{Formally, if the firm chooses $L$ when it is inspected, it
faces a fine $f >0$ and a probability $q \in (0,1)$ of shutting down. One can show that there exists $\underline{f}>0$ and $\underline{q} \in (0,1)$ such that when $f> \underline{f}$ and $q > \underline{q}$, it is a dominant strategy for both types of the firm to choose $a_t=H$ regardless of their discount factors and the equilibrium being played.}

Our reputation result (Theorem \ref{Theorem1}) extends to cases where the distribution of the feasible actions varies exogenously over time or is correlated with the current $\theta$. 
It also extends when $\theta_t$ is drawn from a potentially different set $\Theta_t$ in every period, 
as long as the patient player's payoff is uniformly bounded for all $t \in \mathbb{N}$. This captures 
situations in which
the client's demand varies over time and which of the firm's action benefits the client is known only after the client arrives.

In these situations, 
the advantage of establishing a reputation for honesty is more pronounced. Due to the complicated nature of future payoff environments, it is impractical for the firm to commit to state-contingent action plans. A reputation for honesty  allows the firm to communicate its intended actions after observing the payoff environment, which sidesteps these complications.

\section{Baseline Model}\label{sec2}
Time is discrete, indexed by $t=0,1...$. A long-lived player $1$ (e.g., a seller) with discount factor $\delta$ interacts with an infinite sequence of short-lived player $2$s (e.g., consumers), with $2_t$ denoting the short-lived player in period $t$.
 Player $1$'s action set is $A$ and player $2$'s action set is $B$.

Each period consists of an announcement stage and an action stage.
In period $t$, an i.i.d. random variable $(\theta_t,\omega_t) \in \Theta \times \Omega$ is drawn according to $p \in \Delta (\Theta \times \Omega)$, where $\theta_t  \in \Theta$ affects player $1$'s stage-game payoff (e.g., their cost of supplying high quality) and
$\omega_t \subset A$ is the set of feasible actions, with 
$\Omega \equiv 2^A \backslash \{\varnothing\}$.
Player $1$ privately observes $(\theta_t,\omega_t)$ and 
announces to player $2_t$
that they intend to play action $m_t \in A$.
Players then  simultaneously choose their actions $a_t \in \omega_t$ and $b_t \in B$.

Player $1$'s stage-game payoff is $u_1(\theta_t,a_t,b_t)$ and player $2_t$'s is $u_2(a_t,b_t)$. Importantly, player $2$'s payoff does not depend on $\theta_t$. 
Each player $2$ who arrives after period $t$ can observe $y_t \in Y$, distributed according to $F(\cdot|a_t,m_t)$. A leading example is when 
$y_t$ is the indicator function $\mathbf{1}\{a_t=m_t\}$, that is, future short-run players observe whether the patient  player has kept their word in the past.

Player $1$ has  private information about their type $\gamma \in \{\gamma_h,\gamma_o\}$, which is either \textit{honest} ($\gamma_h$) or \textit{opportunistic} ($\gamma_o$). Both types share the same stage-game payoff function. 
The honest type is restricted (i) to announce an action that is currently available, i.e., $m_t \in \omega_t$, and (ii) to take an action that matches their announcement, i.e., $a_t=m_t$.
The opportunistic type can announce any action (including ones that are not feasible that period)
and can take any action in $\omega_t$ regardless of their announcement. 
Let $\pi_0 \in (0,1)$ be the prior probability of the honest type according to player $2$s' prior belief.

For every $t \in \mathbb{N}$, 
player $2_t$'s private history is $h_2^t \equiv \{y_0,y_1,...,y_{t-1},m_t\}$, with $h_2^t \in \mathcal{H}_2^t$.
Player $2_t$'s strategy is $\sigma_2^t : \mathcal{H}_2^t \rightarrow \Delta (B)$, with $\sigma_2 \equiv (\sigma_2^t)_{t \in \mathbb{N}}$.
Player $1$'s private history in the announcement stage of period $t$ is
\begin{equation*}
    \widehat{h}_1^t \equiv \{\theta_s,\omega_s,m_s,a_s,b_s,y_s\}_{s=0}^{t-1} \bigcup \{\gamma,\omega_t,\theta_t\},
\end{equation*}
with $\widehat{h}_1^t \in \widehat{\mathcal{H}}_1^t$ and $\widehat{\mathcal{H}}_1 \equiv \bigcup_{t=0}^{\infty} \widehat{\mathcal{H}}_1^t$.
Player $1$'s private history in the action stage of period $t$ is
\begin{equation*}
    \widetilde{h}_1^t \equiv \{\theta_s,\omega_s,m_s,a_s,b_s,y_s\}_{s=0}^{t-1} \bigcup \{\gamma,\omega_t,\theta_t,m_t\},
\end{equation*}
with $\widetilde{h}_1^t \in \widetilde{\mathcal{H}}_1^t$ and $\widetilde{\mathcal{H}}_1 \equiv \bigcup_{t=0}^{\infty} \widetilde{\mathcal{H}}_1^t$.
The opportunistic type's strategy is $\sigma_o \equiv (\widehat{\sigma}_{o},\widetilde{\sigma}_o)$, with $\widehat{\sigma}_{o}: \widehat{\mathcal{H}}_1 \rightarrow \Delta (A)$ their strategy to make announcements and
$\widetilde{\sigma}_{o}: \widetilde{\mathcal{H}}_1 \rightarrow \Delta (A)$ their strategy to take actions, subject to a feasibility constraint that the support of $\widetilde{\sigma}_o (\widetilde{h}_1^t)$ is a subset of $\omega_t$.
The honest type's strategy is $\sigma_h \equiv (\widehat{\sigma}_{h},\widetilde{\sigma}_h)$, with $\widehat{\sigma}_{h}: \widehat{\mathcal{H}}_1 \rightarrow \Delta (A)$ their strategy to make announcements and
$\widetilde{\sigma}_{h}: \widetilde{\mathcal{H}}_1 \rightarrow \Delta (A)$ their strategy to take actions, subject to first,
the support of $\widehat{\sigma}_h (\widehat{h}_1^t)$ is a subset of $\omega_t$, and second, their action matches their announcement 
$\widetilde{\sigma}_h (\widetilde{h}_1^t)=m_t$.

A Nash equilibrium (NE) consists of $(\sigma_o,\sigma_h,\sigma_2)$, in which $\sigma_2^t$ maximizes player $2_t$'s stage-game payoff, and every type of player $1$ chooses a strategy that maximizes their discounted average payoff 
$\mathbb{E}\Big[\sum_{t=0}^{\infty} (1-\delta)\delta^t  u_1(\theta_t,a_t,b_t) \Big]$. 
We assume that $\Theta$, $A$, $B$, and $Y$ are finite sets, which together with discounting  of per period payoffs implies that a  Nash equilibrium exists (Fudenberg and Levine 1983). 

\section{Results}\label{sec3}
We show that when the short-run players face a small amount of uncertainty about the feasibility of the patient player's actions, and $y_t$ is informative about whether the patient player has kept their word in period $t$, 
the patient player can secure their expected (pure) Stackelberg payoff in every equilibrium.\footnote{In what follows, we will simply say \textit{Stackelberg action} and \textit{Stackelberg payoff}, with ``pure'' left implicit. }
By contrast, 
the patient player receives a low payoff in some equilibria when all of their actions are feasible in every period.

Recall that $\omega_t \subset A$ is the set of feasible actions in period $t$. 
For every $\varepsilon>0$, we say that player $1$'s action choice is \textit{$\varepsilon$-flexible} if the probability with which $\omega_t=A$ is at least $1-\varepsilon$. 

\begin{Assumption}\label{Ass1}
For every $a \in A$, $\omega_t=\{a\}$ with strictly positive probability. 
\end{Assumption}
Our next assumption requires $y_t$ to be informative about whether player $1$'s action and announcement match. A leading example that satisfies this assumption is
 $y_t =\mathbf{1}\{a_t=m_t\}$.
\begin{Assumption}\label{Ass2}
 If $a=m$ and $a'=m'$, then (i) $F(\cdot|a,m)=F(\cdot|a',m')$, and (ii) $F(\cdot|a,m)$ does not belong to the convex hull of $\{F(\cdot|a'',m'')\}_{a'' \neq m''}$. 
\end{Assumption}
Let $\textrm{BR}_2: \Delta (A) \rightarrow 2^{B} \backslash \{\varnothing\}$ be player $2$'s best reply correspondence. 
In state $\theta \in \Theta$,
player $1$'s Stackelberg payoff is
\begin{equation*}
    v_1^*(\theta) \equiv \max_{a \in A} \Big\{ \min_{b \in \textrm{BR}_2(a)} u_1 (\theta,a,b) \Big\},
\end{equation*}
and their expected Stackelberg payoff is 
 $v_1^* \equiv \sum_{\theta \in \Theta} p(\theta) v_1^*(\theta)$.
\begin{Theorem}\label{Theorem1}
Suppose the environment satisfies Assumptions \ref{Ass1} and \ref{Ass2}. 
For every $\eta>0$, there exist $\underline{\delta} \in (0,1)$ and $\varepsilon>0$ such that when $\delta > \underline{\delta}$ and player $1$'s action choice is $\varepsilon$-flexible, each type of player $1$ receives payoff at least $v_1^*-\eta$ in every Nash equilibrium.\footnote{When $y_t=\mathbf{1}\{a_t=m_t\}$, a patient player $1$ can guarantee payoff approximately $v_1^*$ in every weak rationalizable outcome defined in \cite{watson} in the perturbed game where player $1$ is honest with positive probability.} 
\end{Theorem}
The proof is in Appendix \ref{secA}. 
For some intuition, consider the example in which $y_t=\mathbf{1}\{a_t=m_t\}$.
Let $a^*: \Theta \rightarrow A$ be any mapping such that 
$a^*(\theta) \in \arg\max_{a \in A} \Big\{ \min_{b \in \textrm{BR}_2(a)} u_1 (\theta,a,b) \Big\}$
for every $\theta$.

Fix any equilibrium, and consider the honest type's payoff when they announce $a^*(\theta_t)$ in period $t$ whenever $a^*(\theta_t) \in \omega_t$.
The second part of Assumption \ref{Ass2} implies that whether player $1$'s action coincides with their announcement is informative about their type in the ``bad'' periods where player $2$ fails to best reply to the announcement. Assumption \ref{Ass1} requires that for each $a \in A$, $\omega_t=\{a\}$ with positive probability, which implies that the honest type makes each announcement with positive probability in every period.\footnote{Our reputation result extends when $\omega_t$ is the set of feasible announcements. When the honest type trembles and makes each announcement with positive probability, our result also extends to settings where player $1$ makes their announcement \textit{before} knowing which of their actions are feasible. See Corollary \ref{Cor3} for details.} 
The first part of Assumption \ref{Ass2} guarantees that the honest type's deviation generates the same distribution of player $2$'s histories as the honest type's equilibrium strategy. 

If player $2_t$ fails to best respond to the announced action, player $2$ must assign  a significant probability to  the event that player $1$ is opportunistic and takes $a_t \neq m_t$. Hence, observing $a_t=m_t$ should increase the posterior probability with which player $1$ is honest. Thus, there exists at most a bounded number of periods where player $2$ believes that player $1$ chooses $a_t \neq m_t$ with significant probability. In other words, player $2$ must believe that player $1$ is honest or opportunistic but keeps their word most of the time.


Theorem \ref{Theorem1} requires that the patient player 
can choose any action in $A$ with probability close to $1$. In fact, as long as the environment satisfies Assumptions \ref{Ass1} and \ref{Ass2}, one can establish a reputation result using the same argument as the proof of Theorem \ref{Theorem1} after adjusting the definition of expected Stackelberg payoff. For every $(\theta,\omega) \in \Theta \times \Omega$, let
\begin{equation}
    u_1^{*}(\theta,\omega) \equiv \max_{a \in \omega} \Big\{ \min_{b \in \textrm{BR}_2(a)} u_1 (\theta,a,b) \Big\},
\end{equation}
and let 
\begin{equation*}
    u_1^{*} \equiv \sum_{(\theta,\omega) \in \Theta \times \Omega} p(\theta,\omega) u_1^{*}(\theta,\omega).
\end{equation*}
One can show that when the environment satisfies Assumptions \ref{Ass1} and \ref{Ass2}, a patient player can guarantee payoff $u_1^{*}$ in all equilibria when $\delta$ is close enough to $1$.

Assumption \ref{Ass2} rules out situations in which player $2$s can perfectly observe player $1$'s past announcements, or more generally can observe signals that statistically identify the content of player $1$'s past announcements even when $a_t=m_t$. To study such situations, we extend our result when player $2$s can observe informative signals $z_t$ about the past realizations of $a_t$ and $m_t$, as long as each of them can only observe the realizations of $z_t$ in a bounded number of previous periods.

Formally, let $z_t \in Z$, where $z_t$ is distributed according to
$G(\cdot|m_t,a_t)$. We make no restrictions on $G$ except that its support $Z$ is a finite set. Suppose for every $t \in \mathbb{N}$, player $2_t$ can observe player $1$'s announcement $m_t$, the history  $\{y_0,...,y_{t-1}\}$, as well as a  (possibly stochastic) subset of $\{z_0,...,z_{t-1}\}$ that has at most $K \in \mathbb{N}$ elements. 
Whether player $1$ can observe  $y$ and $z$ are irrelevant for our result.

Our assumption on the asymmetry between player $2$s' observations of $y_t$ and $z_t$ is motivated by retail markets in developing economies. Due to the lack-of record-keeping institutions, detailed information about sellers' actions and announcements (e.g., the quality of their services, various attributes of their products, the content of their advertisements, and so on, which correspond to $z_t$) is likely to get lost over time. By contrast,  coarse information about sellers' records, such as whether they have kept their word (which corresponds to $y_t$), is likely to be more persistent due to  social learning and word-of-mouth communication.
Corollary \ref{Cor1} extends Theorem \ref{Theorem1}, with proof in Appendix C. 
\begin{Corollary}\label{Cor1}
Suppose the environment satisfies Assumptions \ref{Ass1} and \ref{Ass2}, and there exists $K \in \mathbb{N}$ such that each player $2$ observes the past realizations of $z$ in at most $K$ periods.
For every $\eta>0$, there exist $\underline{\delta} \in (0,1)$ and $\varepsilon>0$ such that when $\delta > \underline{\delta}$ and player $1$'s action choice is $\varepsilon$-flexible, each type of player $1$ has  payoff at least $v_1^*-\eta$ in every Nash equilibrium. 
\end{Corollary}

Now we show why uncertainty about which of player 1's actions are feasible
is necessary for Theorem 1 to hold in general. We show this for situations in which $\omega_t=A$ with probability $1$ and $y_t=\mathbf{1}\{a_t=m_t\}$.\footnote{We can construct low-payoff equilibria when $\omega_t=A$ with probability $1$ and player $2_t$ can perfectly observe $\{a_s,m_s\}_{s=0}^{t-1}$.
We do not know how to construct low-payoff equilibria when some actions are not available with positive probability and player $2_t$ can perfectly observe $\{a_s,m_s\}_{s=0}^{t-1}$.}

We start by introducing two auxiliary one-shot games that have the same payoff functions as the original stage  game. The first auxiliary game does not have a communication stage: Player $1$ observes $\theta$, and then players act simultaneously without any communication.
Let $v_1^{min}$
be player $1$'s lowest  Nash equilibrium payoff in this game. The second auxiliary game has an  action recommendation stage: Player $1$ observes $\theta$, makes a  recommendation $\widehat{b} \in B$ to player $2$ before players take their actions.
Let $\widehat{v}_1$ be player $1$'s lowest pure-strategy equilibrium payoff in this game. If there is no pure-strategy equilibrium in this game, let $\widehat{v}_1 =+\infty$.

Let $\mathcal{B}$ be the set of mappings $\beta: A \rightarrow \Delta (B)$ such that 
$\beta(a)$ is a best reply to 
$a$ for every $a \in A$. Abusing notation, 
let $p$ be the distribution of $\theta$. Let 
\begin{equation}\label{5.2}
  v_1' \equiv  \min_{A' \subset A, \beta \in \mathcal{B}} 
    \sum_{\theta \in \Theta} 
    p(\theta) 
    \max_{a \in A'} u_1(\theta,a,\beta(a))
\end{equation}
subject to
\begin{equation}\label{5.3}
    \sum_{\theta \in \Theta} 
    p(\theta) 
    \max_{a \in A'} u_1(\theta,a,\beta(a))
    \geq \min \{ v_1^{min}, \widehat{v}_1\}. 
\end{equation}
Theorem \ref{Theorem2} shows that when all of player $1$'s actions are feasible in every period, there are  equilibria in which both types of player $1$ have  payoff no more than $v_1'$,
\begin{Theorem}\label{Theorem2}
If $\omega_t=\{A\}$ with probability $1$ and $y_t=\mathbf{1}\{a_t=m_t\}$, then there exists $\underline{\delta} \in (0,1)$ such that for every $\delta > \underline{\delta}$, there exists an equilibrium in which both types of player $1$'s payoff is $v_1'$.
\end{Theorem}
The proof of this result and  a subsequent lemma  are in Appendix \ref{secB}. 

In order to understand the connections between the conclusion of Theorem \ref{Theorem2}
and that of Theorem \ref{Theorem1}, 
we compare $v_1'$ with player $1$'s expected Stackelberg payoff $v_1^*$ and their minmax payoff. To start with, one can verify that $v_1^* \geq v_1'$ when players' payoffs are generic and the  auxiliary game without communication admits a pure-strategy equilibrium.\footnote{The generic requirement is that player $1$ has a strict best reply to every $b \in B$ for every $\theta \in \Theta$, and player $2$ has a strict best reply to every $a \in A$. The existence of a pure-strategy equilibrium in the auxiliary game without communication rules out zero-sum games such as matching pennies, where  the patient player cannot benefit from committing to pure actions.} 
Next, we introduce a class of games under which $v_1'$ is strictly less than $v_1^*$, and under an additional supermodularity 
condition, $v_1'$ equals player $1$'s minmax payoff.
\begin{Condition}
There exists a complete order on $A$ such that for every $\theta \in \Theta$, $u_1 (\theta,a,b)$ is strictly decreasing in $a$, and there exists $\theta \in \Theta$ such that player $1$'s Stackelberg action in state $\theta$ is not the lowest element in $A$. 
\end{Condition}
\begin{Lemma}
If every $\theta \in \Theta$ occurs with positive probability and the stage-game payoffs satisfy  supermodularity, then
\begin{enumerate}
    \item $v_1'<v_1^*$.
    \item In addition, if there also exists a complete order on $B$ such that $u_1$ is strictly increasing in $b$ and $u_2$ has strictly increasing differences in $a$ and $b$, then $v_1'$ is player $1$'s minmax payoff.
\end{enumerate}
\end{Lemma}
Condition 1 and the additional assumption in Lemma 1 fit applications such as product choice games, where a firm finds it costly to exert high effort, can strictly benefit from consumers' trust, and can benefit from committing to high effort in states where its production cost is low enough. The consumers have stronger incentives to trust the firm when the latter exerts higher effort. 
Our conditions also apply to  games of entry deterrence (Kreps and Wilson 1982, Milgrom and Roberts 1982), capital taxation (Phelan 2006), monetary policy (Barro 1986), and trust games more generally (Liu and Skrzypacz 2014).

\section{Extensions}\label{sec4}
We note here that the conclusion of Theorem 1 extends to two alternative scenarios.

\paragraph{Announcing Product Quality:} Suppose that players move sequentially in the stage game. In period $t \in\mathbb{N}$, player $1$ (e.g., a firm) chooses their effort $a_t \in A$, privately observes the quality of its product $x_t \in X$ which is distributed according to $g(\cdot|a_t) \in \Delta (A)$, and makes an announcement $m_t \in X$ about quality. Player $2_t$ (e.g., a consumer) observes $m_t$ as well as whether $x_{s}$ coincides with $m_{s}$ for all $s \leq t-1$ before choosing $b_t \in B$. We assume that $A$, $B$, and $X$ are finite sets.

Player $1$ is either an honest type who strategically chooses actions $a_t \in A$ but announces $x_t$ truthfully, or an opportunistic type who strategically chooses both the actions and the announcements.
Both types have stage-game payoff  $u_1(a_t,b_t)$ and discount factor  $\delta\in (0,1)$.  Player $2_t$'s payoff is $u_2(x_t,b_t)$, i.e., their payoff depends only on product quality and their purchasing decision. 
This fits the model of  \cite{jullienEtAl20} except that there is a positive probability of the  honest type, and the ex-post quality is not directly observed by subsequent consumers.  For every $x \in X$,
let $\textrm{BR}_2(x) \subset B$ be the set of pure best replies against $x$. 
Player $1$'s \textit{optimal commitment payoff} is
\begin{equation}\label{4.1}
  v^{**} \equiv   \max_{a \in A}  \Big\{  \sum_{x \in X} g(x|a) \min_{b \in \textrm{BR}_2(x)} u_1(a,b) \Big\}.
\end{equation}
\begin{Corollary}\label{Cor2}
If $g(\cdot|a)$ has full support for every $a \in A$, then for every $\varepsilon>0$, there exists $\underline{\delta} \in (0,1)$ such that when $\delta > \underline{\delta}$, every type of player $1$'s payoff in every Nash equilibrium is at least $v^{**}-\varepsilon$.
\end{Corollary}
The proof is in Appendix D, which uses similar ideas as the proof of Theorem \ref{Theorem1}. 

Next we  show that  reputations for honesty cannot guarantee player $1$ their optimal commitment payoff
when product quality is a perfect signal of player $1$'s effort. Suppose $X=A$ and $g(a|a)=1$ for every $a \in A$, and players' stage-game payoffs are given by the following matrix:
\begin{center}
\begin{tabular}{| c | c | c |}
  \hline
  $-$ & $T$ & $N$ \\
  \hline
  $H$ & $1,2$ & $-1,0$ \\
  \hline
  $L$ & $2,-2$ & $0,0$ \\
  \hline
\end{tabular}
\end{center}
Player $1$'s optimal commitment payoff is $1$, which can be obtained by committing to play $H$. 

We construct a Perfect Bayesian equilibrium in which player $1$'s payoff is $0$. 
On the equilibrium path,
both types of player $1$ play $L$ 
and announce $L$
in every period, and player $2$s play $N$ at every on-path history. After observing announcement $H$, player $2$s believe that player $1$ is the opportunistic type and has played $L$ with probability $1/2$ in the current period, and best reply by playing $N$. This equilibrium survives both when player 2s can only observe whether $m_t$ matches with $x_t$ 
in all previous periods, and when player $2$s can observe the values of $x_t$ and $m_t$ in all previous periods. 

\paragraph{Making Announcements Before Knowing the Set of Feasible Actions:} In some applications, the patient player makes announcements before knowing which of their actions are feasible, and an honest individual may break their word when the action they announced turns out to be infeasible. Theorem 1 extends to this setting if (1) the honest type trembles and makes each announcement with positive probability, and (2) the probability with which all actions are feasible is close to $1$.

In period $t \in \mathbb{N}$, player $1$ observes $\theta_t \in \Theta_t$ and makes an announcement about their intended action $m_t \in A_t$. Player $2_t$ observes $m_t$, player $1$ observes the realization of $\omega_t \in \Omega \equiv 2^A \backslash \{\varnothing\}$, and then both players choose $(a_t,b_t) \in \omega_t \times B$ simultaneously. Future player $2$s observe $y_t \equiv \mathbf{1}\{a_t=m_t\}$. We assume
$\{\omega_t,\theta_t\}_{t \in \mathbb{N}}$ are i.i.d. over time, with
$p \in \Delta (\Omega \times \Theta)$ their joint distribution.

Player $1$ is either an opportunistic type who can take any action regardless of their announcement, or an honest type who chooses $a_t=m_t$ as long as $m_t \in \omega_t$.
Both types of player $1$ tremble when 
making announcements, i.e., there exists $\eta>0$ such that the probability with which each type makes each announcement is at least $\eta$ at every information set.  
\begin{Corollary}\label{Cor3}
For every $\varepsilon>0$, there exist 
$\underline{\delta} \in (0,1)$, $\overline{\eta}>0$, and $c \in (0,1)$, 
such that when $\delta > \underline{\delta}$, $\eta \in (0,\overline{\eta})$ and the probability that $\omega_t=A$ is at least $1-\eta c$, then each type of player $1$ receives payoff at least $v_1^*-\varepsilon$ in every Nash equilibrium.
\end{Corollary}
The proof is in Appendix E. Unlike our baseline model and  reputation models with noisy monitoring such as \cite{fl92}, when the honest type uses the strategy of announcing their Stackelberg action and keeping their word whenever it is feasible,  their reputation may deteriorate in expectation.

Our proof starts by  showing that when $\omega_t=A$ with probability close to $1$, the probability that the honest type keeps their word in equilibrium is close to $1$, 
and for reputation to deteriorate when the honest type keeps their word, 
the opportunistic type must also keep their word with probability close to $1$. It implies that in those periods, player $2$ has a strict incentive to best reply to player $1$'s announcement, and moreover, the amount of reputation deterioration 
is small.
By contrast, in ``bad'' periods where player $2$ has a strict incentive not to best reply against player $1$'s announcement, the probability that the opportunistic type breaks their promise is large and 
keeping one's word leads to a significant improvement in one's reputation. Although the number of bad periods can be unbounded, their fraction goes to zero as the probability of $\omega_t=A$ goes to one.

\appendix
\section{Proof of Theorem 1}\label{secA}
Fix any Nash equilibrium $(\sigma_o,\sigma_h,\sigma_2)$ and consider any history $h^t$ that occurs with strictly positive probability under  $(\sigma_o,\sigma_h,\sigma_2)$.

For $i \in \{h,o\}$, let $P^{\sigma_i,\sigma_2}$ be the probability measure over $Y^{\infty}$ induced by $(\sigma_i,\sigma_2)$. Denote player 2’s belief over player 1’s private history as a function of  $h^t_2$ by $\beta(\hat h^t_1|h^t_2)$, and let $\sigma_o(h^t_2)$ be the expected distribution of opportunistic player 1’s joint announcement-action pairs implied by $\beta(h^t_1|h^t_2)$, with $\hat \sigma_o(h^t_2)$ and $\tilde \sigma_o(h^t_2)$ the marginal distributions of announcements and actions, respectively. Let $\pi_t$ be the probability of the honest type according to player $2$'s belief in period $t$ after observing $\{y_0...,y_{t-1}\}$. According to Bayes rule,
\begin{equation}
\pi_t = \frac{P^{\sigma_h,\sigma_2}(y_0,...,y_{t-1}) \pi_0}{P^{\sigma_h,\sigma_2}(y_0,...,y_{t-1}) \pi_0 + P^{\sigma_o,\sigma_2}(y_0,...,y_{t-1}) (1- \pi_0)}.
\end{equation}
Let 
\[
\alpha_t(m_t) \equiv \pi_t \hat \sigma_h(h^t_2)(m_t) + (1-\pi_t) \hat \sigma_o(h^t_2)(m_t), 
\]
which is the probability of announcement $m_t$ conditional on $h^t_2$. 
Let $\xi_t(m_t)$ be the probability that $a_t=m_t$ conditional on $m_t$,
\[
\xi_t(m_t) \equiv \frac{\pi_t \tilde \sigma_h(h^t_2)(m_t)}{\pi_t \hat \sigma_h(h^t_2)(m_t) + (1-\pi_t) \hat \sigma_o(h^t_2)(m_t)} + \frac{(1-\pi_t) \tilde \sigma_o(h^t_2)(m_t)}{\pi_t \hat \sigma_h(h^t_2)(m_t) + (1-\pi_t) \hat \sigma_o(h^t_2)(m_t)}.
\]
Let $\xi_t$ be the unconditional probability that  player $1$'s action matches their announcement:
\begin{equation}\label{3.2}
\xi_t \equiv \sum_{a \in A} \alpha_t (a) \xi_t(a).
\end{equation}
Let $\underline{\rho} \equiv \min_{a \in A} \Pr(\omega_t=\{a\})$, which by Assumption \ref{Ass1} is strictly positive. 
This implies that $\alpha_t (m)> \underline{\rho}$ for every announcement $m \in A$.
Let $\overline{\lambda} \in (0,1)$ be the smallest real number such that for every $\theta \in \Theta$, player $2$ strictly prefers one of the actions in $\textrm{BR}_2(a^*(\theta))$ to all actions outside of $\textrm{BR}_2(a^*(\theta))$ when they believe that player $1$ plays $a^*(\theta)$ with probability strictly more than $\overline{\lambda}$.

Consider the honest type's payoff when they use strategy $\sigma_h^* \equiv (\widehat{\sigma}_h^*,\widetilde{\sigma}_h^*)$, where $\widetilde{\sigma}_h^* (m)=m$ for every $m \in A$, and $\widehat{\sigma}_h^* (\theta_t,\omega_t)=a^*(\theta_t)$ when $a^*(\theta_t) \in \omega_t$ and is  uniform over the actions in $\omega_t$ when $a^*(\theta_t) \notin \omega_t$.
For any history $h^t$, suppose there exists $m_t \in A$ such that $\xi_t(m_t)\leq \overline{\lambda}$, then $\xi_t \leq \xi^* \equiv 1-(1-\overline{\lambda}) \underline{\rho}$. Let $d(\cdot || \cdot)$ denote the KL-divergence, and let $F^* \equiv F(\cdot|a,a)$.
Let
\begin{equation}\label{3.3}
D^* \equiv \min_{a \neq m} d \Big(F^*\Big\| \xi^* F^* +(1-\xi^*) F(\cdot|a,m)\Big).
\end{equation}
Part 2 of Assumption \ref{Ass2} implies that $D^*>0$ and is independent of player $1$'s discount factor $\delta$.

Part 1 of Assumption \ref{Ass2} implies that 
$P^{\sigma_h,\sigma_2}=P^{\sigma_h^*,\sigma_2}$.  Let $F(y|h^t_2) = \sum_{(a,m) \in A^2} F(y|a,m) \sigma_o(h^t_2)(a,m)$ so that $F(\cdot| h^t_2))$ denotes the distribution over $y_t$ induced by $\sigma_o(h_2^t)$.

Similar to Gossner (2011), the chain rule for relative entropy implies:
\begin{equation*}
\begin{split}
    - \log \pi_0 &\geq d\Big(P^{\sigma_h,\sigma_2}\Big|\Big|\pi_0 P^{\sigma_h,\sigma_2}+(1-\pi_0) P^{\sigma_o,\sigma_2} \Big)\\
    &=d\Big(P^{\sigma_h^*,\sigma_2}\Big|\Big|\pi_0 P^{\sigma_h,\sigma_2}+(1-\pi_0) P^{\sigma_o,\sigma_2} \Big)\\
    &=\mathbb{E}^{(\sigma_h,\sigma_2)}\Big[
    \sum_{t=0}^{\infty}
    d\Big(
    F^*
    \Big|\Big|
    \pi_t F^* + (1-\pi_t)F(\cdot|h^t_2)
    \Big)
      \Big],
\end{split}
\end{equation*}
where the first equality comes from the fact that $(\sigma_h,\sigma_2)$ and $(\sigma_h^*,\sigma_2)$ induce the same distribution over $y$ given the first part of Assumption 2. 

Therefore,
$d(F^*||F(h^t)) \geq D^*$
if $h^t$ is such that $\xi_t(a_t) \leq \overline{\lambda}$ for some $a_t \in A$, so
the expected number of such periods is at most 
\begin{equation}\label{3.5}
   \overline{T}(\pi_0) \equiv \Big\lceil \frac{-\log \pi_0}{D^*}. \Big\rceil
\end{equation}
Hence 
the honest type's payoff from $\sigma_h^*$ is at least
\begin{equation}\label{3.6}
\delta^{T(\pi_0)}  \Big\{    (1-\frac{\varepsilon}{\min_{\theta \in \Theta} p(\theta)})    v_1^*
 +\frac{\varepsilon}{\min_{\theta \in \Theta} p(\theta)} \underline{v}_1 \Big\} + (1-\delta^{T(\pi_0)}) \underline{v}_1,
\end{equation}
in which $\underline{v}_1$ is player $1$'s lowest stage-game payoff.
Expression (\ref{3.6}) converges to $v_1^*$ when $\delta \rightarrow 1$ and $\varepsilon \rightarrow 0$. Since the opportunistic type's payoff is weakly greater than the honest type's payoff, their equilibrium payoff is also weakly more than (\ref{3.6}).

\section{Proofs of Theorem 2 and Lemma 1}\label{secB}
\paragraph{Proof of Theorem 2:} Suppose $(A',\beta)$ 
solves (\ref{5.2}) subject to (\ref{5.3}), and consider the following strategy profile: At every on-path history with $\theta_t=\theta$, both types of player $1$ announce the same $a \in \arg\max_{a \in A'}  u_1(\theta,a,\beta(a))$ and match their actions with their announcements. Player $2$ chooses $\beta(a)$ following announcement $a$, and chooses $\beta(a')$ if player $1$'s announcement does not belong to $A'$, where $a'$ is an arbitrary element of $A'$.
At every $h^t$ where  $y_s \neq 1$ for some $s <t$, player $2$ believes that player $1$ is opportunistic. If $v_1^{min} \leq \widehat{v}_1$, then players coordinate on the worst stage-game Nash equilibrium for player $1$. If $v_1^{min} > \widehat{v}_1$, then players coordinate on the worst pure-strategy Nash equilibirum in the second auxiliary game. Player $2$s' incentive constraints are trivially satisfied, and player $1$'s incentive constraint is implied by (\ref{5.3}). 

\paragraph{Proof of Lemma 1:} 
Let $\underline{a} \equiv \min A$ and let $\underline{b}$ be player $2$'s best reply against $\underline{a}$. 
According to (\ref{5.2}), $v_1' \leq \sum_{\theta \in \Theta} p(\theta) u_1(\theta,\underline{a},\underline{b})$. According to (3.1), $v^* \geq \sum_{\theta \in \Theta} p(\theta) u_1(\theta,\underline{a},\underline{b})$, and the inequality is strict from  the second part of Condition 1 and the fact that each $\theta$ has strictly positive probability.

If $u_1$ is strictly decreasing in $b$ and $u_2$ has strictly increasing differences, then $u_1(\theta,\underline{a},\underline{b})$ is player $1$'s minmax payoff in state $\theta$. Since $v_1' \leq \sum_{\theta \in \Theta} p(\theta) u_1(\theta,\underline{a},\underline{b})$, $v_1'$ is player $1$'s minmax value. 

\section{Proof of Corollary 1}\label{secC} 
Recall the definition of $\xi^*$ in the proof of Theorem 1. Let
\begin{equation}
    \widehat{\xi} \equiv 1- \underline{\rho}^K (1-\xi^*).
\end{equation}
Suppose the honest type announces the Stackelberg action whenever it is available.
By construction, if player $2_t$ believes that $a_t=m_t$ with probability at least $\widehat{\xi}$ after observing $\{y_0,...,y_{t-1}\}$, then because $\omega=\{a\}$ with probability at least $\underline{\rho}$, their posterior belief after observing at most $K$ realizations of $\{z_0,...,z_{t-1}\}$ attaches probability at least $\xi^*$ to $a_t=m_t$. This implies that player $2$ has an incentive to best reply against player $1$'s announced actions. Let
\begin{equation}\label{A.1}
\widehat{D} \equiv \min_{a \neq m} d \Big(F^*\Big\| \widehat{\xi} F^* +(1-\widehat{\xi}) F(\cdot|a,m)\Big).
\end{equation}
The same argument as in the proof of Theorem 1 implies that in expectation, there exist at most 
\begin{equation}\label{3.5}
   \widehat{T}(\pi_0) \equiv \Big\lceil \frac{-\log \pi_0}{\widehat{D}} \Big\rceil
\end{equation}
periods in which player $2$ 
believes that $a_t=m_t$ with probability less than $\widehat{\xi}$ after observing $\{y_0,...,y_{t-1}\}$. As a result, the honest type's payoff is at least $v^*$ when $\varepsilon \rightarrow 0$ and $\delta \rightarrow 1$.

\section{Proof of Corollary 2}\label{secD}
Fix any equilibrium  $(\sigma_o,\sigma_h,\sigma_2)$ and consider any history $h^t$ that occurs with strictly positive probability under $P^{(\sigma_o,\sigma_h,\sigma_2)}$.
Let $\xi_t \in [0,1]$ be the probability player 2 assigns to 
$\mathbf{1}\{m_t=x_t\}$ 
after observing $(y_0,...,y_{t-1})$ but before observing $m_t$. Let $\xi_t(m_t) \in [0,1]$ be the probability 
their belief attaches to $\mathbf{1}\{m_t=x_t\}$ 
after observing $(y_0,...,y_{t-1})$ and $m_t$.
Let $\alpha_t(m_t) \in [0,1]$ be the probability of announcement $m_t$ 
conditional on $(y_0,...,y_{t-1})$. By definition,
\begin{equation}
\xi_t=\sum_{m \in X} \alpha_t (m) \xi_t(m).
\end{equation}
Let $\underline{g} \equiv \min_{(a,x) \in A \times X} g(x|a)$, which is strictly positive under the full support assumption. 
Let $\overline{\lambda} \in (0,1)$ be large enough such that for every $x \in X$, if player $2$ believes that $x$ occurs with probability at least $\overline{\lambda}$, then they strictly prefer one of the actions in $\textrm{BR}_2(x)$ to all actions that do not belong to $\textrm{BR}_2(x)$. Let $\xi^* \equiv 1-\underline{g} (1-\overline{\lambda})$. When
$\xi_t \geq \xi^*$, they believe that $m_t=x_t$ with probability more than $\overline{\lambda}$ after observing $m_t$. As in the proof of Theorem 1, one can show that the expected number of periods in which $\xi_t \leq \xi^*$ is uniformly bounded from above. Therefore, player $1$'s payoff is at least $v^*$ as $\delta \rightarrow 1$. 

\section{Proof of Corollary 3}\label{secE}
For every $t \in \mathbb{N}$, let $\nu_t \in \{0,1\}$ be a random variable such that 
\begin{enumerate}
    \item $\nu_t=1$ if for every $a \in A$, when player $1$ announces $a$ period $t$, player $2$ strictly prefers one of the actions in 
$\textrm{BR}_2(a)$ to all actions that do not belong to $\textrm{BR}_2(a)$,
\item $\nu_t=0$ otherwise.
\end{enumerate}
There exists $\xi \in (0,1)$ such that for every $a \in A$, all actions outside of $\textrm{BR}_2(a)$ are strictly inferior when player $2$ believes that player $1$ plays $a$ with probability more than $\xi$. Let $\pi_t \in \Delta \{\gamma_h,\gamma_o\}$ be player $2$'s belief in period $t$ after observing $\{y_0,...,y_{t-1}\}$ but before observing $m_t$, and let $\widehat{\pi}_t \in \Delta \{\gamma_h,\gamma_o\}$ be player $2$'s belief in period $t$ after observing $\{y_0,...,y_{t-1}\}$ and $m_t$.
Let $l_t \equiv \log \frac{\pi_t(\gamma_h)}{\pi_t(\gamma_o)}$ and $\widehat{l}_t \equiv \log \frac{\widehat{\pi}_t(\gamma_h)}{\widehat{\pi}_t(\gamma_o)}$. Since player $1$ trembles with probability $\eta$, we have $\widehat{l}_t-l_t \geq \log \eta$. 
As a result, there exists $l^* \in \mathbb{R}_+$ such that
$l_t \geq l^*$ implies that $\nu_t=1$.

Let $\Pr(\cdot| \sigma_h)$ be the probability under the honest type's equilibrium strategy. 
Let $\Pr(\cdot|\sigma_o)$ be the probability under the opportunistic type's equilibrium strategy.
Let $\Pr(\cdot| \sigma_h^*)$ be the probability under the strategy of announcing the Stackelberg action in each state, and keeps one's word whenever it is feasible. 
Let $\rho \equiv 1-\Pr(\omega_t=A)$. 
We have
\begin{equation*}
    \Pr(m_t=a_t|\sigma_h) \in [1-\rho , 1-\rho \eta], \quad
    \Pr(m_t=a_t|\sigma_h^*) \in [1-\rho, 1-\rho \eta], \quad \textrm{and} \quad
    \Pr(m_t=a_t|\sigma_o) \leq 1-\rho \eta. 
\end{equation*}
In periods where $\nu_t=0$, the Markov's inequality implies that $\Pr(m_t=a_t|\sigma_o)< 1-\eta(1-\xi)$. 
When player $1$ plays according to $\sigma_h^*$,  Bayes Rule implies that 
\begin{equation}\label{C.1}
    \mathbb{E}[l_{t+1}-l_t | l_t]
    = D( p_t(\sigma_h^*) || p_t(\sigma_o) )
    - D( p_t(\sigma_h^*) || p_t(\sigma_h) ),
\end{equation}
where $D(\cdot||\cdot)$ denoted the KL-divergence, and $p_t(\sigma)$ is the distribution over $y_t$ under strategy $\sigma$. When $\nu_t=1$, we have
\begin{equation}\label{C.2}
   D( p_t(\sigma_h^*) || p_t(\sigma_o) )
    - D( p_t(\sigma_h^*) || p_t(\sigma_h) ) \geq - D(1-\rho || 1-\eta \rho) \equiv -\beta.
\end{equation}
When $\nu_t=0$, we have
\begin{equation}\label{C.3}
   D( p_t(\sigma_h^*) || p_t(\sigma_o) )
    - D( p_t(\sigma_h^*) || p_t(\sigma_h) ) \geq 
  D(1-\rho || 1-\eta (1-\xi)) - D(1-\rho || 1-\eta \rho) \equiv \alpha,
\end{equation}
where $D(x_1||x_2)$ denotes the KL-divergence between a distribution that attaches probability $x_1$ to $y_t=1$, and one that attaches probability $x_2$ to $y_t=1$. 
When $\rho$ is small enough relative to $\eta$, the RHS of (\ref{C.3}) is strictly positive, and moreover, for any fixed $\eta>0$, $\frac{\alpha}{\beta} \rightarrow +\infty$ as $\rho \rightarrow 0$.

We establish a lower bound for the expected value of $\sum_{t=0}^{\infty} (1-\delta)\delta^t \nu_t$ when $\delta$ is close enough to $1$. 
Recall that $y_t=\mathbf{1}\{a_t=m_t\}$.
Let $Z_t$ be a random variable such that
\begin{equation*}
    Z_t=\log \frac{\Pr(y_t|\sigma_h)}{\Pr(y_t|\sigma_o)} \textrm{ with probability } \Pr(y_t|\sigma_h^*) \textrm{ for every } y_t \in Y. 
\end{equation*}
Our analysis above suggests that when $\nu_t=1$, we have $\mathbb{E}[Z_t|\sigma_h^*] \geq -\beta$, and when $\nu_t=0$, we have $\mathbb{E}[Z_t|\sigma_h^*] \geq \alpha$. 
\begin{Claim}
For every $\varepsilon>0$, there exists $T \in \mathbb{N}$ such that for every $t \geq T$, $\mathbb{E}[\sum_{s=0}^{t-1} \nu_s|\sigma_h^*] \geq t(\frac{\alpha}{\alpha+\beta}-\varepsilon)$ with probability more than $1-\varepsilon$. 
\end{Claim}
Our proof uses the Azuma-Hoeffding's inequality:
\begin{Lemma}
Let $\left\{Z_{0},Z_{1},\cdots \right\}$ be a martingale such that $|Z_{k}-Z_{k-1}|\leq c_{k}$.
For every $N \in \mathbb{N}$ and  $\epsilon_1 > 0$, we have
$$\Pr[Z_{N}-Z_{0}\geq \epsilon_1 ]\leq \exp \left(-\frac{\epsilon_1^{2}}{2\sum _{k=1}^{N}c_{k}^{2}}\right).$$
\end{Lemma}
\begin{proof}[Proof of Claim 1:]
Construct a martingale process $\{\widetilde{l}_t\}_{t \in \mathbb{N}}$ recursively. Let $\widetilde{l}_0 \equiv l_0$, and for every $t \in \mathbb{N}$, let $\widetilde{l}_{t+1} \equiv \widetilde{l}_{t}+Z_t-\mathbb{E}[Z_t|\widetilde{l}_t]$. Suppose $\frac{1}{t}\mathbb{E}[\sum_{s=0}^{t-1} \nu_s |\sigma_h^*] \leq \frac{\alpha}{\alpha+\beta}-\varepsilon_1$, then
\begin{equation}
    \sum_{s=0}^{t-1} \mathbb{E}[Z_s|\sigma_h^*] \geq t\varepsilon_1 (\alpha +\beta).
\end{equation}
Therefore,
\begin{equation}\label{C.5}
    \Pr (l_t \leq l^* |\sigma_h^*)
    =\Pr \Big(\widetilde{l}_t-\widetilde{l}_0 \leq l^*-l_0-t\varepsilon_1 (\alpha+\beta)  \Big| \sigma_h^* \Big)
    \leq \exp \Big(-\frac{(l^*-l_0-t\varepsilon_1 (\alpha+\beta))^2}{2 t C^2}\Big),
\end{equation}
where $C$ is the difference between the largest realization of $Z_t$ and the smallest realization of $Z_t$. The RHS of (\ref{C.5}) vanishes to zero exponentially as $t \rightarrow +\infty$. Since $\nu_t=1$ when $l_t \geq l^*$, we know that for every $\varepsilon_0>0$, there exists $T_1 \in \mathbb{N}$, such that for every $t \geq T_1$,  
$\mathbb{E}[\sum_{s=0}^{t-1} \nu_s |\sigma_h^*] \leq t \Big( \frac{\alpha}{\alpha+\beta}-\varepsilon_1 \Big)$
implies that $\nu_t=1$ with probability at least $1-\varepsilon_0$, and by setting $\epsilon_0 < \frac{\beta}{\alpha+\beta}$ we obtain in this case that
$\mathbb{E}[\sum_{s=0}^{t} \nu_s|\sigma_h^*]-\mathbb{E}[\sum_{s=0}^{t-1} \nu_s|\sigma_h^*] \geq 1-\varepsilon_0 > \frac{\alpha}{\alpha+\beta}$. Then for every $t > T_1$, one can show by induction that $\mathbb{E}[\sum_{s=0}^{t-1} \nu_s|\sigma_h^*] \ge (t-T_1-1)(\frac{\alpha}{\alpha + \beta} - \varepsilon_1)$, and the claim follows by choosing any $\epsilon > \epsilon_1 > 0$, and $T \ge T_1 + 1 + \frac{\alpha}{(\epsilon - \epsilon_1)(\alpha + \beta)}$.
\end{proof}
We use the following well-known equation:
\begin{equation}\label{C.6}
 \mathbb{E}[   \sum_{t=0}^{\infty} (1-\delta)\delta^t \nu_t
 |\sigma_h^*]
 =(1-\delta)^2 \sum_{t=0}^{+\infty} \delta^t \underbrace{\sum_{s=0}^{t} \mathbb{E}[\nu_s|\sigma_h^*]}_{\geq (\frac{\alpha}{\alpha+\beta}-\varepsilon)(t+1)\textrm{ with probability close to 1}}.
\end{equation}
Claim 1 and (\ref{C.6}) imply that for every $\widehat{\varepsilon} >0$, there exists $\overline{\delta} \in (0,1)$, such that for every $\delta \in (\overline{\delta},1)$, we have
\begin{equation}\label{C.4}
\mathbb{E}[\sum_{t=0}^{\infty} (1-\delta)\delta^t \nu_t |\sigma_h^*] \geq \frac{\alpha}{\alpha+\beta}-\widehat{\varepsilon}.
\end{equation}
Recall that $v^*$ is player $1$'s expected pure Stackelberg payoff. Without loss of generality, we normalize player $1$'s worst stage-game payoff to $0$. If $\nu_t=1$, then the honest type's expected payoff from
announcing their Stackelberg action in every state is at least $(1-\rho-\eta) v^*$. If $\nu_t=0$, then
the honest type's expected payoff is at least $0$. Pick $\overline{\delta} \in (0,1)$ such that $\mathbb{E}[\sum_{t=0}^{\infty} (1-\delta)\delta^t \nu_t] \geq \frac{\alpha}{\alpha+\beta}-\widehat{\varepsilon}$,
and pick $\rho$ small enough such that $\frac{\alpha}{\alpha+\beta}$ is greater than $1-\widehat{\varepsilon}$, the honest type's payoff is at least 
\begin{equation*}
   (1-2 \widehat{\varepsilon}) (1-\eta-\rho) v^*.
\end{equation*}
There exists $c \in (0,1)$ such that the above expression  is greater than $v^*-\varepsilon$ when $\eta$ and $\widehat{\varepsilon}$ are small enough, and $\rho \leq c \eta$.

\end{spacing}
\newpage
\bibliography{bib}
\nocite{*}

\end{document}